\documentclass[12pt,draftcls,onecolumn,english]{IEEEtran}  

\usepackage{amsmath,amsfonts, amssymb,color}

\newtheorem{theorem}{Theorem}
\newtheorem{corollary}{Corollary}
\newtheorem{lemma}{Lemma}
\newtheorem{remark}{Remark}
\newtheorem{definition}{Definition}
\newtheorem{proposition}{Proposition}
\newtheorem{example}{Example}
\newtheorem{assumption}{Assumption}

\usepackage{epsfig}
\usepackage{psfrag}

\title{\LARGE \bf Robustness of Information Diffusion Algorithms to Locally Bounded Adversaries}

\author{Haotian~Zhang~and~Shreyas~Sundaram%
\thanks{This material is based upon work supported in part by the Natural Sciences and Engineering Research Council of Canada, and by the Waterloo Institute for Complexity and Innovation.}
\thanks{The authors are with the Department of Electrical and Computer Engineering at the University of Waterloo. E-mail for corresponding author: {\tt ssundara@uwaterloo.ca}.}
}

\begin{document}

\maketitle
\thispagestyle{empty}
\pagestyle{empty}

\begin{abstract}
We consider the problem of diffusing information in networks that contain malicious nodes.  We assume that each normal node in the network has no knowledge of the network topology other than an upper bound on the number of malicious nodes in its neighborhood.  We introduce a topological property known as $r$-robustness of a graph, and show that this property provides improved bounds on tolerating malicious behavior, in comparison to traditional concepts such as connectivity and minimum degree. We use this topological property to analyze the canonical problems of distributed consensus and broadcast, and provide sufficient conditions for these operations to succeed.  Finally, we provide a construction for $r$-robust graphs and show that the common preferential-attachment model for scale-free networks produces a robust graph.
\end{abstract}


\section{Introduction}
A core question in the study of large networks (both natural and engineered) is: how do the actions of a small subset of the population affect the global behavior of the network?
For instance, the fields of sociology and epidemiology examine the spread of ideas, decisions and diseases through populations of people, based on the patterns of contact between the individuals in the population \cite{Morris00,Newman02,Easley10}.  In this context, one can ask whether a few stubborn individuals (who do not change their beliefs) are able to affect the decisions reached by the rest of the population \cite{Stauffer05,Xie11}.  Similarly, the efficacy of engineered networks (such as communication networks, or multi-agent systems) is often predicated on their ability to disseminate information throughout the network \cite{Hromkovic05,Akyildiz02}.  For example, the `broadcast' operation is used as a building block for more complex functions, allowing certain nodes to inform all other nodes of pertinent information \cite{Hromkovic05}.  Another important operation is that of `distributed consensus', where every node in the network has some information to share with the others, and the entire network must come to an agreement on an appropriate function of that information \cite{Lynch96,Olfati07,Tsitsiklis84,Jadbabaie03,SundaramHadjicostis07}.

The ability of a few individuals to affect the global behavior of the system is clearly a double-edged sword.  When the network contains legitimate leaders or experts, it is beneficial to ensure that the innovations introduced by these small groups spread throughout the population.  On the other hand, networks that facilitate diffusion are also vulnerable to disruption by individuals that are not acting with the best interests of the society in mind.  In engineering applications, these individuals could correspond to faulty or malicious nodes that do not follow preprogrammed strategies due to malfunctions or attacks, respectively.  Thus, a fundamental challenge is to identify network properties and diffusion dynamics that allow legitimate information to propagate throughout the network, while limiting the effects of illegitimate individuals and actions.

The problem of transmitting information over networks (and specifically, reaching consensus) in the presence of faulty or malicious nodes has been studied
extensively over the past several decades (e.g., see \cite{Lynch96,Hromkovic05} and the references therein).
It has been shown that if the connectivity of the network is $2f$ or less for some nonnegative integer $f$, then $f$ malicious nodes can conspire to prevent some of the nodes from correctly receiving the information of other nodes in the network.
Conversely, when the network connectivity is $2f+1$ or higher, there are various algorithms to allow reliable dissemination of information (under the wireless broadcasting model of communication) \cite{SundaramHadjicostis11,Pasqualetti11}.  However, these methods require that all nodes have full knowledge of the network topology, along with the specific parameters of the algorithm applied by all other nodes.  Furthermore, the computational overhead for these methods is generally quite high \cite{Lynch96,SundaramHadjicostis11}.

It is not surprising that there is a tradeoff between how much each node knows about the overall network and the conditions required for those nodes to overcome malicious adversaries.  The objective of this paper is to analyze information dissemination strategies in networks with adversaries when each normal node only has access to its neighbors' values, and does not know anything about the rest of the network (i.e., the topology, number of nodes, location and behavior of malicious nodes, etc.); it only knows that the total number of adversaries in its own neighborhood is bounded by some known quantity.  We introduce the concept of {\it $r$-robust} graphs, and show that such graphs provide resilience to malicious nodes.  We focus on the particular applications of fault-tolerant broadcast and distributed consensus, and similarly to \cite{Pelc05}, we consider a {\it locally bounded} fault model where there is an upper bound on the number of adversarial nodes in the neighborhood of any reliable node, but there is no other {\it a priori} bound on the total number of adversaries in the network.
In the case of fault-tolerant broadcast, our conditions can be applied to show that broadcast will succeed in certain networks that do not meet the conditions provided in \cite{Pelc05}.  For distributed consensus, our conditions provide separate sufficient and necessary conditions for all normal nodes to reach consensus while limiting the ability of locally-bounded malicious nodes to influence the final value.


\section{System Model}
Consider a network modeled by the directed graph $\mathcal{G}=\{\mathcal{V},\mathcal{E}\}$, where $\mathcal{V}=\{1,...,n\}$ is the set of nodes and $\mathcal{E} \subseteq \mathcal{V} \times \mathcal{V}$ is the set of edges in the network. An edge $(j,i)\in \mathcal{E}$ indicates that node $i$ can be influenced by (or receive information from) node $j$.
The set of neighbors of node $i$ is defined as $\mathcal{V}_i=\{j\mid (j,i)\in \mathcal{E}\}$ and the degree of $x_i$ is denoted by $\deg_i =\rvert \mathcal{V}_i \rvert$.

Suppose that each node in the network starts out with some private information (an opinion, a vote, a measurement, etc.).  We will model each piece of information as a real number, and denote node $i$'s initial information as $x_i[0]$.
Further suppose that the network is synchronous and at each time-step $t \in \mathbb{N}$, each node updates its information based on its interactions with its neighbors.  We will model these updates as
$$
x_i[t+1] = f_i(\{x_{j}[t]\}), \enspace j \in \mathcal{V}_i \cup \{i\},
$$
where 
$f_i(\cdot)$ can be an arbitrary function (and perhaps different for each node, depending on its role in the network).  We assume that each $f_i(\cdot)$ is specified {\it a priori} for each node $i$ in order to achieve some pre-specified global objective.  However, we also allow for the possibility that certain nodes in the network {\it do not} follow their prescribed strategy. We will use the following definitions in this paper.

\begin{definition}
A node $i$ is said to be {\it normal} if it applies $f_i(\cdot)$ at every time-step $t$, and it is called {\it malicious} otherwise.  Denote the set of malicious nodes by $\mathcal{M}$, and the set of normal nodes by $\mathcal{N} = \mathcal{V}\setminus\mathcal{M}$.
\end{definition}

Note that comparing with the {\it Byzantine fault model} \cite{Dolev86},  the fault model considered here does not allow the malicious (or normal) nodes to transmit different values at each time step (i.e., every pair of nodes will receive the same values from their common neighbors at each time step).  This assumption is natural in many network realizations (such as wireless networks), and the above definition allows the malicious nodes to behave in an arbitrary manner under this communication modality. However, we will also show that many of the results in this paper also apply to the Byzantine fault model (where a Byzantine mode can send arbitrary and different values to different neighbors at each time-step). 

Clearly, there is no hope of achieving any objective if every node in the network is malicious.  Instead, it is reasonable to consider the resilience of the network to various specific classes of malicious nodes.  For instance, a common assumption in the literature on fault-tolerant distributed algorithms is that the total number of malicious nodes in the network is upper bounded by some number $f$ \cite{Lynch96,Hromkovic05,SundaramHadjicostis11}, i.e., the {\it $f$-total malicious model}.  In very large networks, however, it may be the case that the total number of failures or adversaries is quite large.  To capture this, we will consider in this paper a {\it locally bounded} fault model, taken from \cite{CYKoo2004, Pelc05}.

\begin{definition}[$f$-local set]
A set $\mathcal{S} \subset \mathcal{V}$ is {\it $f$-local} if it contains at most $f$ nodes in the neighborhood of the other nodes, i.e., $\rvert \mathcal{V}_i\bigcap \mathcal{S}\rvert \le f$, $\forall i\in \mathcal{V}\setminus\mathcal{S}$.
\label{def:f_local}
\end{definition}

\begin{definition}[$f$-local malicious model]
A set $\mathcal{M}$ of malicious nodes is {\it $f$-locally bounded} if it is a $f$-local set.
\label{def:f_local_mal}
\end{definition}

Note that the $f$-total malicious model can be regarded as a special case of the $f$-local malicious model. In the rest of the paper, we will focus on two specific algorithms, (i) distributed consensus, and (ii) broadcast, and derive topological conditions that guarantee resilience to locally bounded adversaries.


\section{Asymptotic Consensus with Locally Bounded Adversaries}
\label{sec:asymp_con}

The use of linear iterative strategies for facilitating distributed consensus has attracted significant attention in the control community (see \cite{Olfati07} and references therein).
In such strategies, at each time-step $t \in \mathbb{N}$, each node communicates with its neighbors and updates its local value as
\begin{equation*}
x_i[t+1]=w_{ii}[t]x_i[t]+\sum_{j\in \mathcal{V}_i}w_{ij}[t]x_j[t],
\end{equation*}
where $w_{ij}[t]$ is the weight assigned to node $j$'s value by node $i$ at time-step $t$.

\begin{definition}[Asymptotic Consensus]
The system is said to reach {\it asymptotic consensus} if $\lvert x_i[t]-x_j[t]\rvert \rightarrow 0$ as $t\to\infty$, for all $i, j \in \mathcal{V}$.
\end{definition}

Various conditions have been provided in the literature that will guarantee that all nodes in the network reach asymptotic consensus \cite{Xiao04,Ren05,Moreau05,Jadbabaie03,Tsitsiklis84} (we will discuss some of these results in greater detail later in this paper).  It is typical in these existing works to assume the following conditions for the weights.

\begin{assumption}
\label{assu:weights}
There exists a real-valued constant $\alpha \in(0, 1)$ such that
\begin{itemize}
\item $w_{ii}[t]\ge\alpha,\forall i,t$
\item $w_{ij}[t]=0$ if  $j\not\in \mathcal{V}_i, \forall i,j,t$
\item $w_{ij}[t] \ge\alpha$ if $j\in \mathcal{V}_i,\forall i,j,t$
\item $\sum_{j=1}^{n}w_{ij}[t]=1, \forall i,t$.
\end{itemize}
\end{assumption}

The lower bound on the weights is imposed to guarantee convergence; there are various examples of graphs and updates with no lower bounds for which consensus does not occur \cite{Lorenz10}.  

Here, we are interested in the case where not all nodes in the network apply the above linear iterative strategy.  Instead, these malicious nodes can update their values in {\it arbitrary} ways to prevent or bias the consensus value in the network.  We now review some recent results pertaining to this scenario.

\subsection{Previous Results on Resilience of Asymptotic Consensus}
The paper \cite{Jadbabaie03} studied the use of linear iterative strategies as a mechanism for achieving {\it flocking} behavior in multi-agent systems.  They showed that if a `leader' node in the network does not update its value at each time-step (i.e., it maintains a constant value), then the above linear iterative strategy (where every other node updates its value to be a convex combination of its neighborhood values) will cause all nodes to asymptotically converge to the value of the leader.  While this may be acceptable behavior when the network has a legitimate leader, it also seems to indicate that a simple asymptotic consensus scheme can be easily disrupted by just a single malicious node.
A similar analysis was done in \cite{Gupta06}, where it was argued that since the asymptotic consensus scheme can be disrupted by a single node that maintains a constant value, it can also be disrupted by a single node that updates its values arbitrarily (since maintaining a constant value is a special case of arbitrary updates).  Both of these works only considered a straightforward application of the linear iteration for asymptotic consensus, without having the normal nodes perform any operations to avoid the influence of malicious behavior.

In \cite{SundaramHadjicostis11}, the authors provided a comprehensive analysis of linear iterative strategies in the presence of malicious nodes.  They demonstrated that linear iterative strategies are able to achieve the minimum bound required to disseminate information reliably; specifically, when a network is $2f+1$ connected, $f$ malicious nodes will be unable to prevent any node from calculating any function of the initial values (under the broadcast model of communication).  This result was extended in \cite{Pasqualetti11} to analyze linear iterative strategies for asymptotic consensus in the presence of faulty agents (in addition to malicious agents), and \cite{Teixeira10} studied the problem of detecting attacks in networks of linear continuous-time systems.
While these results require minimal connectivity, they also require each normal node to have full knowledge of the network topology, along with strong computational and storage capabilities.  The paper \cite{Chapman10} considers the problem of reducing the influence of external intruders on asymptotic consensus in tree networks.  They propose a rewiring scheme whereby each node changes its parent node in an effort to slow down the effect of externally connected adversaries.  While the approach presented in that paper is distributed, it only applies to tree topologies and requires that the location and intention of the adversaries to be known by the nodes.

In \cite{Dolev86}, the authors introduced the {\it Approximate Byzantine Consensus} problem, in which the normal nodes are required to achieve approximate agreement (i.e., they should converge to a relatively small convex hull contained in their initial values) in the presence of the $f$-total Byzantine faults in finite time.\footnote{If the network is synchronous, and if one allow $t\to\infty$, then approximate agreement is equivalent to the asymptotic consensus problem considered in this paper.} To solve this problem in complete networks (where there is a direct connection between every pair of nodes), they proposed the following algorithm: each node disregards the largest and smallest $f$ nodes in the network and updates its state to be the average of a carefully chosen subset of the remaining values. They proved that the algorithm achieves approximate agreement in synchronous and asynchronous networks if there are more than $3f$ and $5f$ nodes in the network, respectively, and provided a provable convergence rate for both networks. The algorithm was extended to be a family of algorithms, named the {\it Mean-Subsequence-Reduced} or {\it MSR} algorithms, in \cite{Azad94}. Although the research on Approximate Byzantine Consensus for complete networks is mature, there are only a few papers that study this problem in general network topologies \cite{Azad93, Azad01, Azad02}; furthermore, these works have only provided conditions for local convergence (convergence of a subset of nodes) \cite{Azad93, Azad02}, or for global convergence in special network topologies \cite{Azad01}.

The recent paper \cite{LeBlanc11} proposes a continuous-time variation of the MSR algorithms, named the {\it Adversarial Robust Consensus Protocol (ARC-P)}, to solve asymptotic consensus under the $f$-total malicious model.
The authors show that the limit of the state trajectory of each normal node exists and in complete networks, the normal nodes asymptotically reach consensus on a value that is in the interval formed by their initial states. In \cite{LeBlanc_LowCompResConsAdv_HSCC12}, the authors extend the results from \cite{LeBlanc11} to slightly more general networks and provide sufficient conditions in terms of traditional graph metrics, such as the in-degree and out-degree of nodes in the network.
However, we will show in this paper that these traditional metrics (such as degree and connectivity) studied in \cite{Lynch96,SundaramHadjicostis11,LeBlanc11,LeBlanc_LowCompResConsAdv_HSCC12} are no longer the key factors that determine the efficacy of algorithms that make purely local filtering decisions.  Instead, we develop a novel topological condition for general networks, termed $r$-robustness, which we show to be much more fundamental in characterizing the behavior of algorithms such as MSR (including ARC-P) and fault-tolerant broadcast.

\subsection{Description of the Algorithm}
We consider the network $\mathcal{G} = \{\mathcal{V},\mathcal{E}\}$, where at each time-step, each node $i$ receives the values of the nodes in $\mathcal{V}_i$.  Node $i$ does not know which, if any, nodes in its neighborhood are malicious; it only knows that there are at most $f$ malicious nodes.  In order to limit the influence of any malicious nodes, each normal node disregards the most extreme values from its neighborhood at each time-step, and uses the remaining values in its linear update.  More formally, we extend the MSR algorithm to be the {\it Weighted-Mean-Subsequence-Reduced (W-MSR)} algorithm as follows.

\begin{enumerate}
\item At each time-step $t$, each normal node $i$ receives values from all of its neighbors, and ranks them from largest to smallest.
\item If there are $f$ or more values larger than $x_i[t]$, normal node $i$ removes the $f$ largest values.  If there are fewer than $f$ values larger than $x_i[t]$, normal node $i$ removes all of these larger values.  This same logic is applied to the smallest values in normal node $i$'s neighborhood.
Let $\mathcal{R}_i[t]$ denote the set of nodes whose values were removed by normal node $i$ at time-step $t$.
\item Each normal node $i$ updates its value as
\begin{equation}
x_i[t+1]=w_{ii}[t]x_i[t]+\sum_{j\in \mathcal{V}_i\setminus\mathcal{R}_i[t]}w_{ij}[t]x_j[t],
\label{eqn:ft_update}
\end{equation}
where the weights $w_{ii}[t]$ and $w_{ij}[t]$ satisfy Assumption~\ref{assu:weights}.
\end{enumerate}

\begin{remark}
Note that the algorithm is the same in time-varying networks, except that $\mathcal{V}_i$ is a function of $t$.
Furthermore, note that the above algorithm essentially falls within the class of MSR algorithms, with the following generalizations. First, we allow arbitrary time-varying weights on the edges at each time-step, subject to the constraints listed in Assumption~\ref{assu:weights}. Note that \cite{Azad08} also proposed an extension of the MSR algorithm which allows convex time-invariant weights.  Second, we allow a node to throw away fewer than $2f$ values if its own value falls within the extreme range, thereby allowing it to take full advantage of the available information. Moreover, we will analyze this algorithm in arbitrary graph topologies (not only fully-connected ones).
\end{remark}

We call the largest number of values that each node could throw away the {\it parameter} of the algorithm (it is equal to $2f$ in the above algorithm).  Note that the set of nodes disregarded by node $i$ can change over time, depending on their relative values.  Thus, even the network topology itself is fixed, the algorithm effectively mimics a time-varying network.  In other words, one can view this as a consensus algorithm with state-dependent switching.

\begin{remark}
Consensus algorithms with state-dependent switching have drawn increased attention in recent years \cite{HK2002,Ts2009}.
For example, the following model was introduced in \cite{HK2002} to capture {\it opinion dynamics} in networks:
$$
x_{i}[t+1]=\frac{\sum_{j:\rvert x_{i}[t]-x_{j}[t]\rvert <1}x_{j}[t]}{\lvert \{j:\rvert x_{i}[t]-x_{j}[t] \rvert <1\}\rvert}.
$$
The constraint $\rvert x_{i}[t]-x_{j}[t]\rvert <1$ represents `bounded confidence' among these nodes: an agent considers one of its neighbors' opinions as reasonable and accepts it if their opinions differ by less than $1$.  There are various differences in the analysis in these previous works in comparison with this paper. First, the above updating scheme assumes that the underlying graph is complete, so that each node sees all other nodes and selects only those whose values are close to its own.
Second, there exists a fixed threshold (1 in the above scheme) to represent `bounded confidence', and this might cause the agents to converge to different clusters for certain choices of initial states \cite{Ts2009}.  Most importantly, these previous works on state-dependent connectivity do not consider the presence of malicious nodes; we posit that the fixed threshold in the update rule still allows a malicious node to draw all of the other nodes to any desired consensus value, simply by waiting until all node values have converged sufficiently close together, and then slowly inducing drift by keeping its value near the edge of the fixed threshold. The algorithm considered in this paper, on the other hand, applies to general topologies and inherently limits the amount of bias than can be introduced by any $f$-local set of malicious nodes.
\end{remark}

Note that the W-MSR algorithm is efficient, scalable and fully distributed. The algorithm needs very limited computation and storage, which is especially important in resource limited networks.  Furthermore, no node needs to know the topology of the network;
the only requirement is that each normal node knows (or assumes) an upper bound of $f$ for the maximum number of malicious nodes in its neighborhood.  Due to this  simplicity, it is perhaps unreasonable to expect that this algorithm will be able to completely eliminate the effects of all malicious nodes.
Instead, as in \cite{LeBlanc11}, we will seek to ensure that the algorithm is $f$-local safe, which we define below.

\begin{definition}[$f$-local safe]
Given the network $\mathcal{G}$, let $\mathcal{M}$ be the set of malicious nodes (satisfying the $f$-local property) and $\mathcal{N}$ be the set of normal nodes. The W-MSR algorithm is said to be {\it $f$-local safe} if all normal nodes asymptotically reach consensus for any choice of initial values, and the consensus value is in the range $[m_{\mathcal{N}}[0], M_{\mathcal{N}}[0]]$, where $M_{\mathcal{N}}[0]$ and $m_{\mathcal{N}}[0]$ denote the largest
and smallest initial values of the normal nodes, respectively.
\end{definition}

Note that the above definition does not say that the malicious nodes will have no influence on the final consensus value.  It only says that a $f$-local set of malicious nodes should not be able to bias the consensus value to be something outside the range of normal initial values.  There are various practical applications where this is useful.  For instance, consider a large sensor network where every sensor takes a measurement of its environment, captured as a real number.  Suppose that at the time of measurement, all values taken by correct sensors fall within a range $[a,b]$, and that all sensors are required to come to an agreement on a common measurement value.  If the range of measurements taken by the normal sensors is relatively small, it will likely be the case that reaching agreement on a value within that range will form a reasonable estimate of the measurements taken by all sensors.  However, if a set of malicious nodes is capable of biasing the consensus value to be outside this range, the functioning of the network could be severely disrupted.

\begin{remark}
Note that our concept of $f$-local safe holds even if a $f$-local set of malicious nodes change their {\it initial values}; no matter what these malicious nodes change their value to be, if the algorithm achieves consensus, it will be on a value that is in the range of the initial values of the normal nodes.
\end{remark}

\subsection{Analysis of the Algorithm}

Denote the maximum and minimum values of a set $\mathcal{S}$ of nodes at time step $t$ as $M_{\mathcal{S}}[t]$ and $m_{\mathcal{S}}[t]$, respectively, i.e., $M_{\mathcal{S}}[t]=\max\{x_i[t]\mid i\in\mathcal{S}\}$ and $m_{\mathcal{S}}[t]=\min\{x_i[t]\mid i\in\mathcal{S}\}$. Further denote $\Phi[t]=M_{\mathcal{N}}[t]-m_{\mathcal{N}}[t]$. Note that $\Phi[t]\to 0$ as $t\to \infty$ if and only if the normal nodes reach asymptotic consensus.

\begin{lemma}
Under the $f$-local malicious model, if the normal nodes reach consensus, the W-MSR algorithm is $f$-local safe.
\label{lem:valid}
\end{lemma}

\begin{proof}
At each time-step $t \in \mathbb{N}$, after receiving values from its neighbors, each normal node throws away at most $f$ largest and $f$ smallest values.  Since there are at most $f$ malicious nodes in the neighborhood of any normal node, the remaining values must be in the range $[m_{\mathcal{N}}[t], M_{\mathcal{N}}[t]]$; if all of the malicious nodes were removed, then only the normal nodes are left, and if some malicious nodes' values are adopted, then the malicious nodes must have had values inside the range of the normal values.  Since the update~\eqref{eqn:ft_update} is a convex combination of these values, we have $x_i[t+1] \in [m_{\mathcal{N}}[t], M_{\mathcal{N}}[t]]$ for all $t$ (showing that $m_{\mathcal{N}}[t]$ is nondecreasing and $M_{\mathcal{N}}[t]$ is non-increasing).  If all normal nodes reach consensus, it must be that $\Phi[t] \rightarrow 0$, indicating that $x[t] \rightarrow m_{\mathcal{N}}[t]$ (or $M_{\mathcal{N}}[t]$), and thus the result follows by virtue of the fact that these quantities are monotonic.
\end{proof}

The task now is to provide conditions under which the normal nodes reach consensus, despite the (arbitrary) actions of the malicious nodes.  Recall that when there are up to $f$ malicious nodes in the entire network, and each normal node knows the entire network topology (along with the weights used by all other nodes), a network connectivity of $2f+1$ is necessary and sufficient to overcome the malicious nodes \cite{SundaramHadjicostis11}.  The first question that comes to mind is thus: what does the connectivity of the network have to say about the ability of the algorithm to facilitate consensus?  Unfortunately, the following result shows that there exist graphs with large connectivity, but that fail to reach consensus under this algorithm.


\begin{proposition}
There exists a network with connectivity $\kappa = \lfloor \frac{n}{2}\rfloor +f -1 $ that cannot achieve asymptotic consensus using the W-MSR algorithm with parameter $2f$.
\label{prop:counter_example}
\end{proposition}
\begin{proof}
Construct an undirected graph as follows. Take two fully-connected graphs of $\lfloor \frac{n}{2}\rfloor$ and $\lceil \frac{n}{2}\rceil$ nodes, respectively, and call these sets $\mathcal{A}$ and $\mathcal{B}$. Number nodes in $\mathcal{A}$ and $\mathcal{B}$ as $a_1,a_2,\dots,a_{\lfloor \frac{n}{2}\rfloor}$ and $b_1,b_2,\dots,b_{\lceil \frac{n}{2}\rceil}$, respectively.
When $n$ is even (singular), for any node $a_i\in\mathcal{A}$, if $i\le\rvert \mathcal{B}\rvert-f+1$ ($\rvert \mathcal{B}\rvert-f+2$), connect $a_i$ with nodes $b_i, b_{i+1},\dots, b_{i+f-1}$; otherwise, connect $a_i$ with nodes $b_i, \dots, b_{\lceil \frac{n}{2}\rceil}$ and nodes $b_1, \dots, b_{i+f-\lceil \frac{n}{2}\rceil-1}$. Form similar connections for nodes in $\mathcal{B}$. Then each node in $\mathcal{A}$ has exactly $f$ neighbors in $\mathcal{B}$, each node in $\mathcal{B}$ has at most $f$ neighbors in set $\mathcal{A}$.

Next we will prove that the connectivity of this graph is $\lfloor \frac{n}{2}\rfloor+f-1$. Let $\mathcal{C}=\{\mathcal{C}_{\mathcal{A}}, \mathcal{C}_{\mathcal{B}}\}$ be a vertex cut, where $\mathcal{C}_{\mathcal{A}}=\mathcal{C}\cap\mathcal{A}$ and $\mathcal{C}_{\mathcal{B}}=\mathcal{C}\cap\mathcal{B}$.
Without loss of generality, assume that $\mathcal{C}_{\mathcal{A}}=\{a_1, a_2,\dots, a_{\rvert \mathcal{C}_{\mathcal{A}} \rvert}\}$; other ways of choosing $\mathcal{C}_{\mathcal{A}}$ are equivalent to this situation by renumbering the nodes. By the definition of a vertex cut, we know $\rvert \mathcal{C}_{\mathcal{A}} \rvert >f$; otherwise, each node in $\mathcal{B}\setminus\mathcal{C}_{\mathcal{B}}$ still has at least one neighbor in $\mathcal{A}$, and since $\mathcal{A}\setminus\mathcal{C}_{\mathcal{A}}$ and $\mathcal{B}\setminus\mathcal{C}_{\mathcal{B}}$ each induce fully-connected subgraphs, we see that the graph will be connected (contradicting the fact that $\mathcal{C}$ is a vertex cut).
When $f< \hspace{0.1cm}\rvert\mathcal{C}_{\mathcal{A}}\rvert\hspace{-0.05cm}<\hspace{-0.05cm}\lfloor\frac{n}{2}\rfloor$, the remaining nodes of $\mathcal{A}$ still have $k= \lfloor \frac{n}{2}\rfloor- \hspace{0.05cm}\rvert\mathcal{C}_{\mathcal{A}}\rvert +f-1$ neighbors in $\mathcal{B}$, which implies we need to remove at least $k$ nodes from $\mathcal{B}$ to disconnect the graph. When $\mathcal{C}_{\mathcal{A}}=\mathcal{A}$, since $\mathcal{B}$ is complete, we know $\rvert\mathcal{C}_{\mathcal{B}}\rvert = \lceil \frac{n}{2}\rceil -1$. Thus the connectivity of this graph is $\lfloor \frac{n}{2}\rfloor+f-1$.

In this graph, assume that all nodes in $\mathcal{A}$ have initial value $c_1$, and all nodes in $\mathcal{B}$ have initial value $c_2$, where $c_1 < c_2$. When any node $a_i$ applies the W-MSR algorithm, all of its $f$ neighbors in $\mathcal{B}$ have the highest values in $a_i$'s neighborhood, and thus they are all disregarded.  Similarly, all of $b_i$'s neighbors in $\mathcal{A}$ are all disregarded as well.
Thus, each node in each set only uses the values from its own set, and no node ever changes its value, which shows that consensus will never be reached in this network.
\end{proof}

Note that the above network also has minimum degree $\lfloor \frac{n}{2}\rfloor + f -1$.  Thus, even networks with a large degree or connectivity are not sufficient to guarantee that the normal nodes will reach consensus, indicating that these metrics are not particularly useful on their own to analyze the performance of this algorithm.  In the next section, we define a topological notion that we term {\it robustness}, and show that this notion more readily characterizes the situations where the algorithm is $f$-local safe.

\section{Robust Graphs}
Taking a closer look at the graph constructed in Proposition~\ref{prop:counter_example}, we see that the reason for the failure of consensus in this case is that no node has enough neighbors in the opposite set; this causes each node to throw away all useful information from the opposite set, and prevents consensus.  Based on this intuition, we define the following property of a set of nodes, which we will show to be key to characterizing the behavior of local filtering algorithms such as W-MSR.

\begin{definition}[$r$-reachable set]
For a graph $\mathcal{G}$ and a subset $\mathcal{S}$ of nodes of $\mathcal{G}$, we say $\mathcal{S}$ is an {\it $r$-reachable set} if $\exists i\in \mathcal{S}$ such that $\rvert \mathcal{V}_i\setminus\mathcal{S}\rvert \ge r$, where $r\in\mathbb{N}_+$.
\end{definition}

In words, a set $\mathcal{S}$ is $r$-reachable if it contains a node that has at least $r$ neighbors outside.
The following lemma follows directly from the definition of an $r$-reachable set.

\begin{lemma}
Consider a graph $\mathcal{G} = \{\mathcal{V},\mathcal{E}\}$, and suppose that $\mathcal{S} \subset \mathcal{V}$ is an $r$-reachable set of nodes.  Then,
\begin{itemize}
\item $\mathcal{S}$ is $r'$-reachable for any $r$\ satisfying 1$\le r' \le r$.
\item If we remove up to $K$ incoming edges of each node $i \in \mathcal{V}$, where $K < r$, then $\mathcal{S}$ is ($r-K$)-reachable.
\end{itemize}
\label{lem:reachable_props}
\end{lemma}

\begin{definition}[$r$-robust graph]
A graph $\mathcal{G}$ is {\it $r$-robust} if for every pair of nonempty, disjoint subsets of $\mathcal{V}$, at least one of the subsets is $r$-reachable.
\label{def:r_robust}
\end{definition}

Based on these definitions, we obtain the following properties of $r$-robust graphs.

\begin{lemma}
For an $r$-robust graph $\mathcal{G}$, let $\mathcal{G}'$ be the graph produced by removing up to $K$ incoming edges of each node in $\mathcal{G}$ ($K < r$).  Then $\mathcal{G}'$ is ($r-K$)-robust.
\label{lem:robust_props}
\end{lemma}

\begin{proof}
First note that the minimum degree of an $r$-robust graph must be at least $r$; otherwise, if there is a node $i$ with degree less than $r$, then by taking the two subsets $\{i\}$ and $\mathcal{V}\setminus\{i\}$, we see that neither subset would have a node with $r$ neighbors outside. Thus, it is possible to remove $K$ incoming edges from every node.  We can now apply the second property in Lemma~\ref{lem:reachable_props} to obtain the desired result.
\end{proof}

\begin{lemma}
If $\mathcal{G}$ is $r$-robust for some $r \ge 1$, then it has a spanning tree.
\label{lem:robust_spanning_tree}
\end{lemma}

\begin{proof}
It is sufficient to show that a $1$-robust graph has a spanning tree.  Consider the $1$-robust graph $\mathcal{G}$.  We will prove that this graph has a spanning tree by contradiction: assume that $\mathcal{G}$ does not have a spanning tree. Decompose the graph into its strongly connected components, and note that since the graph does not have a spanning tree, there must be at least two components that have no incoming edges from any other components.  However, this contradicts the assumption that $\mathcal{G}$ is $1$-robust (at least one of the two subsets must have a neighbor outside the set), so it must be true that there exists a spanning tree.
\end{proof}

When $r = 1$, the above proof is a more direct version of the proof of Theorem 5 in \cite{Moreau05}.



\subsection{Consensus With Locally Bounded Faults}
\label{sec:con_results}

In this subsection, we will explore sufficient and necessary conditions under which the algorithm is $f$-local safe. We first define some notation.

Denote the set of {\it normal} nodes with maximum and minimum values at time step $t$ as $\mathcal{S}_{max}^{\mathcal{N}}[t]$ and $\mathcal{S}_{min}^{\mathcal{N}}[t]$, respectively, i.e. $\mathcal{S}_{max}^{\mathcal{N}}[t]=\{i\mid x_i[t]=M_{\mathcal{N}}[t],i\in\mathcal{N}\}$ and $\mathcal{S}_{min}^{\mathcal{N}}[t]=\{i\mid x_i[t]=m_{\mathcal{N}}[t],i\in\mathcal{N}\}$.

\begin{definition}
For a network $\mathcal{G}$, define the {\it normal network} of $\mathcal{G}$, denoted by $\mathcal{G}_{\mathcal{N}}$, as the network induced by the normal nodes,  i.e., $\mathcal{G}_{\mathcal{N}}=\{\mathcal{N}, \mathcal{E}_{\mathcal{N}}\}$, where $\mathcal{E}_{\mathcal{N}}$ is the set of edges among the normal nodes.
\end{definition}

The following lemma provides a key sufficient condition for the normal nodes to reach consensus.

\begin{lemma}
Under the $f$-local malicious model, the W-MSR algorithm with parameter $2f$ is $f$-local safe if the normal network $\mathcal{G}_{\mathcal{N}}$ of the network is $(f+1)$-robust.
\label{lem:f_local_sufficient}
\end{lemma}

The proof of this result is given in the Appendix.  With the above Lemma in hand, we are now in place to provide a condition on the original network $\mathcal{G}$ that will guarantee that the algorithm is $f$-local safe.

\begin{theorem}
Under the $f$-local malicious model, the W-MSR algorithm with parameter $2f$ is $f$-local safe if the network $\mathcal{G}$ is $(2f+1)$-robust.
\label{thm:local_sufficient}
\end{theorem}
\begin{proof}
By the definition of the normal network, $\mathcal{G}_{\mathcal{N}}$ is obtained by removing up to $f$ incoming edges from each normal node in $\mathcal{G}$. By Lemma~\ref{lem:robust_props}, if $\mathcal{G}$ is $(2f+1)$-robust, then $\mathcal{G}_{\mathcal{N}}$ is $(f+1)$-robust. Finally, by Lemma~\ref{lem:f_local_sufficient}, we get the result.
\end{proof}

The following proposition shows that the $(2f+1)$-robust condition is tight.

\begin{proposition}
For every $f>0$, there exists a $2f$-robust network which fails to reach consensus using the W-MSR algorithm with parameter $2f$.
\label{prop:tight_suff}
\end{proposition}
\begin{proof}
We will prove the result by giving a construction of such a graph, visualized in Figure~\ref{fig:tight}. In Figure~\ref{fig:tight}, $\mathcal{S}_1$, $\mathcal{S}_2$ and $\mathcal{S}_3$ are all complete components with $\rvert \mathcal{S}_1 \rvert =2f$, $\rvert \mathcal{S}_2 \rvert =4f$ and $\rvert \mathcal{S}_3 \rvert =2f$. Each node in $\mathcal{S}_1$ connects to $2f$ nodes of $\mathcal{S}_2$ and each node in $\mathcal{S}_3$ connects to the other $2f$ nodes of $\mathcal{S}_2$, and all these connections are undirected. Node $a$ has incoming edges from all nodes in $\mathcal{S}_1$ and similarly node $b$ has incoming edges from all nodes in $\mathcal{S}_3$.
This is an example of a graph that arises from the construction that we derive in Section~\ref{sec:construction}, where we show that such a graph will be $2f$ robust.
We choose $f$ nodes of $\mathcal{S}_1$ and also $f$ nodes of $\mathcal{S}_3$ to be malicious; note that this constitutes an $f$-local set of malicious nodes. Then we assign node $a$ with initial value $m$, node $b$ with initial value $M$ and the other normal nodes with initial values $c$, such that $m<c<M$. Malicious nodes in $\mathcal{S}_1$ and $\mathcal{S}_3$ will keep their values unchanged at $m$ and $M$, respectively. We can see that, by using the W-MSR algorithm, the values of nodes $a$ and $b$ will never change and thus consensus can not be reached, completing the proof.
\end{proof}

\begin{figure}[Tight]
\centering
\includegraphics[width=7.5cm]{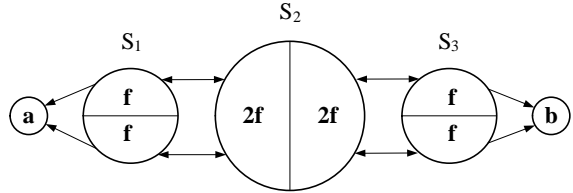}
\caption{Illustration of Proposition \ref{prop:tight_suff}}
\label{fig:tight}
\end{figure}

While the above discussions have been for an underlying time-invariant network $\mathcal{G}$, it is relatively straightforward (albeit notationally tedious) to extend the results to time-varying networks as follows.

\begin{corollary}
Let $\mathcal{G}[t]\hspace{-0.1cm}=\hspace{-0.1cm}\{\mathcal{V}, \mathcal{E}[t]\}$ be a time-varying network with node set $\mathcal{V}$ and edge set at time-step $t$ given by $\mathcal{E}[t]$.  Let $\{t_k \}$ be the set of time-steps when $\mathcal{G}[t]$ is $(2f+1)$-robust. Under the $f$-local malicious model, the W-MSR algorithm with parameter $2f$ is $f$-local safe if $\rvert \{t_k \}\rvert =\infty$ and $\rvert t_{k+1} - t_k\rvert \le c$, $\forall k$, where $c\in\mathbb{N}_+$ is some constant.
\label{cor:time_varying}
\end{corollary}

The proof is similar to Theorem~\ref{thm:local_sufficient}, and we omit it here.


Finally, the following result provides a necessary condition for the W-MSR algorithm to be $f$-local safe.

\begin{theorem}
Under the $f$-local malicious model, the necessary condition for the W-MSR algorithm with parameter $2f$ to be $f$-local safe is that the network $\mathcal{G}$ is $(f+1)$-robust.
\label{thm:necess}
\end{theorem}
\begin{proof}
If the network is not $(f+1)$-robust, there exist two disjoint subsets of nodes that are not $(f+1)$-reachable, i.e., each node in these two sets would have at most $f$ neighbors outside the set. If we assign the maximum and minimum values in the network to these two sets, respectively, the nodes in these sets would never use any values from outside their own sets.  Thus, their values would remain unchanged, and consensus will not be reached.
\end{proof}

Note that the network constructed in Proposition~\ref{prop:counter_example} is only $f$-robust (but not $(f+1)$-robust), since no nodes in sets $\mathcal{A}$ or $\mathcal{B}$ have $f+1$ neighbors outside those sets. Furthermore, it is of interest to note that the derivations of Theorem~\ref{thm:local_sufficient}, Corollary~\ref{cor:time_varying} and Theorem~\ref{thm:necess} did not rely on the fact that malicious nodes send the same value to all their neighbors.  Thus, these results also apply to the $f$-local {\it Byzantine} model of adversaries.

\section{Broadcasting with Locally Bounded Adversaries}
Having characterized the behavior of the consensus algorithm in terms of the $r$-robust property of graphs, we now turn our attention to another important objective in networks: broadcasting a single value throughout the network.  We focus on the following problem, studied in \cite{CYKoo2004,Pelc05}.  Consider a time-invariant communication network $\mathcal{G}=\{\mathcal{V}, \mathcal{E}\}$, with a specially designated source node $s \in \mathcal{V}$.  The source has a value $x_s[0]$ that it wishes to broadcast to every other node in the network.  However, there may be various malicious nodes scattered throughout the network that wish to prevent certain nodes from obtaining the correct value of the source.  The authors consider the set of malicious nodes to be $f$-locally bounded.  To achieve broadcast (i.e., all normal nodes receive the source's message), \cite{CYKoo2004} proposes the following so-called {\it Certified Propagation Algorithm (CPA)}.

\begin{enumerate}
\item At time-step $0$, the source broadcasts its value to all of its neighbors, and maintains its value for all subsequent time-steps.
\item At time-step $1$, all normal neighbors of the source receive the source's value and broadcast it to all of their neighbors.  The normal neighbors of the source maintain this value for all subsequent time-steps.
\item At each time-step $t$, if a normal node has received an identical value from $f+1$ neighbors, then it accepts that value and broadcasts it to all of its neighbors.  This normal node keeps this value for all subsequent time-steps.
\end{enumerate}

Due to the assumption of $f$-locally bounded malicious nodes, it is easy to see that a normal node will only ever accept a value if it is the actual value broadcast by the source.  For {\it CPA}, the following result from \cite{Pelc05} provides a sufficient condition for all normal nodes in the network to eventually accept the value broadcast by the source.

\begin{theorem}[\cite{Pelc05}]
For a graph $\mathcal{G}=\{\mathcal{V},\mathcal{E}\}$ and nodes $v,s\in \mathcal{V}$, let $X(v,s)$ denotes the number of nodes that are in $v$'s neighborhood and are closer to $s$ than $v$. Let $X(\mathcal{G})=\min\{X(v,s)\rvert v,s\in \mathcal{V}, (v,s)\notin \mathcal{E}\}$.  Then CPA succeeds if $X(\mathcal{G}) > 2f$.
\label{thm:pelc_cpa}
\end{theorem}

This is only a sufficient condition; we will now provide a different sufficient condition for CPA to succeed, in terms of the robust-graph property that we have defined.  We will first introduce a variation of the concept of an $r$-robust graph.

\begin{definition}[strongly $r$-robust graph]
For a positive integer $r$, graph $\mathcal{G}=\{ \mathcal{V}, \mathcal{E}\}$ is {\it strongly $r$-robust} if for any nonempty subset $\mathcal{S}\subseteq\mathcal{V}$, either $\mathcal{S}$ is $r$-reachable or there exists a node $i\in \mathcal{S}$ such that $\mathcal{V}\setminus\mathcal{S} \subseteq \mathcal{V}_i$.
\end{definition}

Note that the difference between a strongly $r$-robust graph and the standard $r$-robust graph is that the former requires {\it every} subset of nodes to be either $r$-reachable, or have a node that connects to every node outside the set, whereas the latter only requires that one of any two sets satisfies the property of being $r$-reachable.  Any strongly $r$-robust graph is $r$-robust, but not vice versa.

\begin{theorem}
Under the $f$-local malicious model, CPA succeeds for any source if the network is strongly $(2f+1)$-robust.
\label{thm:CPA}
\end{theorem}

\begin{proof}
All normal neighbors of the source receive the message directly, and thus they all accept it.  We will use contradiction to prove that all other nodes receive the broadcast message.  Suppose that CPA fails to deliver the message to all normal nodes, and let $\mathcal{S}$ denote the set of all such normal nodes.  By the definition of a strongly $(2f+1)$-robust graph, we know that some node $i$ in $\mathcal{S}$ must have $2f+1$ neighbors outside $\mathcal{S}$ or connects to all nodes outside.  For the former situation, at most $f$ of these nodes can be malicious, and all other nodes are normal nodes that have received the message and re-broadcasted it; for the latter, this node would directly connect to the source and thus get the message.  In either case, this contradicts the assumption that node $i$ would fail to get the message, and thus the algorithm achieves broadcast.
 \end{proof}

Note that if the condition of either Theorem~\ref{thm:pelc_cpa} or Theorem~\ref{thm:CPA} is satisfied, CPA will also succeed under the $f$-local Byzantine model. Finally, the following Proposition shows that CPA succeeds in certain networks which do not satisfy the condition proposed in Theorem~\ref{thm:pelc_cpa}.

\begin{proposition}
For some $f$, there exist graphs with $X(\mathcal{G}) \le 2f$ but that are strongly $(2f+1)$-robust.
\end{proposition}

\begin{proof}
For $f=1$, construct an undirected graph $\mathcal{G}$ as follows. Start with a fully-connected graph of five nodes, denoted as $1,2,\dots,5$ in turns. Add two nodes $6$ and $7$ and connect them to nodes $2,3,4$ and $3,4,5$ respectively. Finally, add a node $8$ and connect it to nodes $3,4,6,7$. If we take node $1$ as the source, it's easy to check that in the neighborhood of node $8$, there are only two nodes that are closer to the source. Thus $X(\mathcal{G})\le 2f$ here, but the graph is still strongly $(2f+1)$-robust, and CPA will succeed.
\end{proof}


\section{Constructing an $r$-robust Graph}
\label{sec:construction}

Note that the concept of an $r$-robust graph requires that every possible subset of nodes in the graph satisfies the property of being $r$-reachable.  Currently, we do not have a computationally efficient method to check whether this property holds for an arbitrary graph.  In this section, however, we describe how to construct $r$-robust graphs, and show that our construction contains the preferential-attachment model of scale-free networks as a special case \cite{Albert02}.

\begin{theorem}
Let $\mathcal{G}=\{\mathcal{V},\mathcal{E}\}$ be an $r$-robust graph. Then graph $\mathcal{G}'=\{\{\mathcal{V},v_{new}\},\{\mathcal{E},\mathcal{E}_{new}\}\}$, where $v_{new}$ is a new vertex added to $\mathcal{G}$ and $\mathcal{E}_{new}$ is the edge set related to $v_{new}$, is $r$-robust if $\deg_{v_{new}} \ge r$.
\label{thm:design robust}
\end{theorem}

\begin{proof}
When we take a pair of nonempty, disjoint subsets of nodes from $\mathcal{G}'$, there are two cases. If one of the subsets contains only $v_{new}$, then this subset is $r$-reachable (since $v_{new}$ has $r$ neighbors in $\mathcal{G}'$).  If both of the subsets contain nodes from the original graph $\mathcal{G}$, then at least one of the two sets is $r$-reachable, because these two sets (minus $v_{new}$) exist in the original $r$-robust graph $\mathcal{G}$, and thus one of the sets has a node that has at least $r$ neighbors outside.  Thus, $\mathcal{G}'$ is $r$-robust.
\end{proof}

The above theorem indicates that to build an $r$-robust graph with $n$ nodes (where $n \ge r$), we can start with an $r$-robust graph of order less than $n$ (such as some complete graph), and continually add new nodes with incoming edges from at least $r$ nodes in the existing graph.  The theorem does not specify {\it which} existing nodes should be chosen as neighbors.  A particularly interesting case is when the nodes are selected with a probability proportional to the number of edges that they already have; this is known as {\it preferential-attachment}, and leads to the formation of so-called {\it scale-free} networks \cite{Albert02}.  This mechanism is cited as a plausible mechanism for the formation of many real-world complex networks, and thus our analysis indicates that these networks will also be resilient to locally-bounded malicious nodes (provided that $r$ is sufficiently large when the network is forming).


\section{Conclusion and Discussion}
\label{sec:disc}
We have studied the problem of disseminating information in networks that contain malicious nodes, and where each normal node has no knowledge of the global topology of the network.  We showed that the classical notions of connectivity and minimum degree are not particularly useful in characterizing the behavior of a class of algorithms that relies on purely local filtering rules.  We then introduced the notion of an $r$-robust graph, and showed that this concept allows us to provide conditions for achieving the objectives of distributed consensus and fault-tolerant broadcast, without requiring any knowledge of the graph topology on the part of the nodes in the network.

For distributed consensus, variations and extensions of the approach used in this paper have recently appeared in \cite{HiCoNS12} (for the $f$-total model of malicious, but non-Byzantine, behavior), and in \cite{Nitin12} (for the $f$-total model of Byzantine behavior). The sufficient and necessary conditions proposed in \cite{Nitin12} for the MSR algorithms to achieve consensus also apply for the $f$-local Byzantine model; however, the proof of the necessary condition in \cite{Nitin12} does not apply for the $f$-local {\it malicious} model (which is the scenario considered in this paper), and thus obtaining a single necessary and sufficient condition for consensus under this model is an open problem.  
It is also of interest to note that the notion of an $r$-reachable set is similar to the notion of `clusters', which are topological structures identified in \cite{Easley10} as being impediments to information cascades in networks.  While the topic of information cascades is closely tied to the problems that we consider in this paper, the presence of malicious nodes in our setup significantly complicates the analysis.  Nevertheless, a closer connection to the results in those works is the subject of ongoing research.
Finally, it will be of interest to relate the $r$-robust property defined in this paper to other recent characterizations of network topologies that facilitate fault-tolerant broadcast \cite{Ichimura10}.


\section*{Acknowledgements}
The authors thank Heath J. LeBlanc and Xenofon Koutsoukos for helpful discussions.  The authors also like thank Nitin Vaidya et al. for pointing out related research on Approximate Byzantine Consensus. 


\appendix

\begin{proof}[Lemma~\ref{lem:f_local_sufficient}]
Recall that $\mathcal{N}$ is the set of normal nodes, and define $N = |\mathcal{N}|$.  Let $a_M[t]$ and $a_m[t]$ denote the maximum and minimum value of the normal nodes at time-step $t$, respectively.  From Lemma~\ref{lem:valid}, we know that both $a_M[t]$ and $a_m[t]$ are monotone and bounded functions of $t$ and thus each of them has some limit, denoted by $A_M$ and $A_m$, respectively.   Note that if $A_M=A_m$, the normal nodes will reach consensus.  We will now prove by contradiction that this must be the case.

Suppose that $A_M \ne A_m$ (note that $A_M > A_m$ by definition).   We can then define some constant $\epsilon_0 > 0$ such that $A_M - \epsilon_0 > A_m + \epsilon_0$.  At any time-step $t$ and for any positive real number $\epsilon_i$, let $\mathcal{X}_M(t,\epsilon_i)$ denote the set of all normal nodes that have values in the range $(A_M-\epsilon_i,A_M+\epsilon_i)$, and let $\mathcal{X}_m(t,\epsilon_i)$ denote the set of all normal nodes that have values in the range $(A_m-\epsilon_i,A_m+\epsilon_i)$.  Note that $\mathcal{X}_M(t,\epsilon_0)$ and $\mathcal{X}_m(t,\epsilon_0)$ are disjoint, by the definition of $\epsilon_0$. 

For some $\epsilon$ (which we will show how to choose later) satisfying $\epsilon_0 > \epsilon > 0$, let $t_{\epsilon}$ be such that $a_M[t] <  A_M + \epsilon$ and $a_m[t] >  A_m - \epsilon$, $\forall t\ge t_{\epsilon}$ (we know that such a $t_{\epsilon}$ exists by the definition of convergence).  Consider the disjoint sets $\mathcal{X}_M(t_{\epsilon},\epsilon_0)$ and $\mathcal{X}_m(t_{\epsilon},\epsilon_0)$.  At least one of these two sets must be $(f+1)$-reachable in $\mathcal{G}_\mathcal{N}$ due to the assumption of $(f+1)$-robustness.  If $\mathcal{X}_M(t_{\epsilon},\epsilon_0)$ is $(f+1)$-reachable, there exists some node $x_i \in \mathcal{X}_M(t_{\epsilon},\epsilon_0)$ that has at least $f+1$ normal neighbors outside $\mathcal{X}_M(t_{\epsilon},\epsilon_0)$.  By definition, all of these neighbors have values at most equal to $A_M-\epsilon_0$, and at least one of these values will be used by $x_i$ (since $x_i$ removes at most $f$ values lower than its own value). Note that at each time step, every normal node's value is a convex combination of its own value and the values it uses from its neighbors, and each coefficient in the combination is lower bounded by $\alpha$. Since the largest value that $x_i$ will use at time-step $t_{\epsilon}$ is $a_M[t_{\epsilon}]$, placing the largest possible weight on $a_M[t_{\epsilon}]$ initiates the following sequence of inequalities. 
\begin{align}
x_{i}[t_{\epsilon}+1]
&\le (1-\alpha)a_M[t_{\epsilon}] + \alpha (A_M-\epsilon_0) \nonumber \\
&\le (1-\alpha)(A_M +\epsilon) +\alpha(A_M -\epsilon_0) \nonumber \\
&\le A_M - \alpha\epsilon_0 + (1-\alpha)\epsilon \nonumber.
\end{align}
Note that this upper bound also applies to the updated value of any normal node that is not in $\mathcal{X}_M(t_{\epsilon},\epsilon_0)$, because such a node will use its own value in its update.  Similarly, if $\mathcal{X}_m(t_{\epsilon},\epsilon_0)$ is $(f+1)$-reachable, there exists some node $x_j \in \mathcal{X}_m(t_{\epsilon},\epsilon_0)$ that will satisfy
\begin{equation*}
x_j[t_{\epsilon}+1] \ge A_m + \alpha\epsilon_0 - (1-\alpha)\epsilon \label{eqn:lower_bound_prop}.
\end{equation*}
Again, any normal node that is not in $\mathcal{X}_m(t_{\epsilon},\epsilon_0)$ will have the same lower bound.  Define 
$$
\epsilon_1 = \alpha\epsilon_0 - (1-\alpha)\epsilon,
$$
and consider the sets $\mathcal{X}_M(t_\epsilon+1,\epsilon_1)$ and $\mathcal{X}_m(t_\epsilon+1,\epsilon_1)$.  Since at least one of the sets $\mathcal{X}_M(t_{\epsilon},\epsilon_0)$ and $\mathcal{X}_m(t_{\epsilon},\epsilon_0)$ was $(f+1)$-reachable, it must be that either $|\mathcal{X}_M(t_{\epsilon}+1,\epsilon_1)| < |\mathcal{X}_M(t_{\epsilon},\epsilon_0)|$ or $|\mathcal{X}_m(t_{\epsilon}+1,\epsilon_1)| < |\mathcal{X}_m(t_{\epsilon},\epsilon_0)|$, or both.  Further note that $\epsilon_1 < \epsilon_0$, and thus $\mathcal{X}_M(t_{\epsilon}+1,\epsilon_1)$ and $\mathcal{X}_m(t_{\epsilon}+1,\epsilon_1)$ are still disjoint.

We can repeat this analysis for time-steps $t_{\epsilon} + j$, $j \ge 2$, to define sets $\mathcal{X}_M(t_{\epsilon}+j,\epsilon_j)$ and $\mathcal{X}_m(t_{\epsilon}+j,\epsilon_j)$, where $\epsilon_j$ is defined recursively as $\epsilon_j = \alpha\epsilon_{j-1} - (1-\alpha)\epsilon$.  Furthermore, at time-step $t_\epsilon+j$, either $|\mathcal{X}_M(t_{\epsilon}+j,\epsilon_j)| < |\mathcal{X}_M(t_{\epsilon}+j-1,\epsilon_{j-1})|$ or $|\mathcal{X}_m(t_{\epsilon}+j,\epsilon_j)| < |\mathcal{X}_m(t_{\epsilon}+j-1,\epsilon_{j-1})|$, or both.  Since $|\mathcal{X}_M(t_{\epsilon},\epsilon_0)| + |\mathcal{X}_m(t_{\epsilon},\epsilon_0)| \le N$, there must be some time-step $t_\epsilon+T$ (where $T \le N$) where either $\mathcal{X}_M(t_{\epsilon}+T,\epsilon_T)$ or $\mathcal{X}_m(t_{\epsilon}+T,\epsilon_T)$ is empty.  In the former case, all nodes in the network at time-step $t_\epsilon+T$ have value less than $A_M - \epsilon_T$, and in the latter case all nodes in the network at time-step $t_\epsilon+T$ have value greater than $A_m + \epsilon_T$.  We will show that $\epsilon_T > 0$, which will contradict the fact that the largest value monotonically converges to $A_M$ (in the former case) or that the smallest value monotonically converges to $A_m$ (in the latter case).  To do this, note that
\begin{align*}
\epsilon_T &= \alpha\epsilon_{T-1} - (1-\alpha)\epsilon \\
&= \alpha^2\epsilon_{T-2} - \alpha(1-\alpha)\epsilon - (1-\alpha)\epsilon \\
&\enspace\vdots \\
&= \alpha^T\epsilon_{0} - (1-\alpha)(1+\alpha + \cdots + \alpha^{T-1})\epsilon \\
&= \alpha^T\epsilon_{0} - (1-\alpha^{T})\epsilon\\
&\ge \alpha^{N}\epsilon_0 - (1-\alpha^N)\epsilon.
\end{align*}
If we choose $\epsilon < \frac{\alpha^N}{1-\alpha^N}\epsilon_0$, we obtain $\epsilon_T > 0$, providing the desired contradiction.  It must thus be the case that $\epsilon_0 = 0$, proving that $A_M = A_m$.
\end{proof}

\bibliographystyle{IEEEtran}
\bibliography{refs}
\end{document}